\documentclass[11pt,letterpaper,preprintnumbers,superscriptaddress]{revtex4}
\usepackage{amsfonts,amsmath,amsthm,amssymb,amscd,dsfont}
\usepackage{color,graphicx,tikz,array,accents,hyperref}

\normalfont
\usepackage[T1]{fontenc}

\newtheorem{theorem}{Theorem}
\newtheorem{lemma}[theorem]{Lemma}
\newtheorem{proposition}[theorem]{Proposition}
\newtheorem{corollary}[theorem]{Corollary}

\begin{document}

\title{Morphisms in categories of nonlocal games}
\author{Brad Lackey}
\affiliation{Joint Center for Quantum Information and Computer Science}
\affiliation{University of Maryland Institute for Advanced Computer Studies}
\affiliation{Department of Computer Science, University of Maryland, College Park}
\affiliation{Department of Mathematics, University of Maryland, College Park}
\affiliation{The Cryptography Office, Mathematics Research Group, National Security Agency}
\author{Nishant Rodrigues}
\affiliation{Joint Center for Quantum Information and Computer Science}
\affiliation{Department of Computer Science, University of Maryland, College Park}

\date{\today}

\newcommand{\bra}[1]{\ensuremath \langle{#1}|}%
\newcommand{\ket}[1]{{\ensuremath |{#1}\rangle}}%
\newcommand{\bracket}[2]{\ensuremath \langle{#1}|{#2}\rangle}%
\newcommand{\h}[1]{\ensuremath \mathfrak{#1}}%
\newcommand{\indicator}[1]{\ensuremath \openone_{\{#1\}}}%
\newcommand{\id}{\ensuremath \mathrm{id}}%
\newcommand{\tr}{\ensuremath{\mathrm{tr}}}
\newcommand{\ip}[2]{\ensuremath \langle{#1}|{#2}\rangle}%
\newcommand{\category}[1]{\ensuremath{\mathsf{#1}}}
\newcommand{\Hom}{\mathrm{Hom}}

\renewcommand{\Pr}{\ensuremath \mathrm{Pr}}

\newcommand{\covec}[1]{\ensuremath \undertilde{#1}}
\newcommand{\bcl}[1]{{\color{orange} (\textbf{Brad}: #1)}}

\begin{abstract}
Synchronous correlations provide a class of nonlocal games that behave like functions between finite sets. In this work we examine categories whose morphisms are games with synchronous classical, quantum, or general nonsignaling correlations. In particular, we characterize when morphisms in these categories are monic, epic, sections, or retractions.
\end{abstract}

\maketitle

\begin{section}{Introduction}

Synchronous correlations provide a class of nonlocal games that behave like functions between finite sets. These have long been used to generalize notions such the chromatic number of a graph \cite{galliard2002pseudo, cleve2004consequences, cameron2007quantum, scarpa2012kochen, mancinska2013new, mancinska2016quantum, helton2017algebras, gribling2018bounds, mancinska2018oddities}, or more generally graph homomorphisms \cite{atserias2016quantum, ortiz2016quantum, lupini2017nonlocal, mancinska2017relaxations, musto2018morita}. Recently synchronous correlations in particular have provided a narrow focus where Tsirl'son's problem and the Connes embedding conjecture may be resolved \cite{dykema2016synchronous, harris2017unitary, kim2018synchronous}; specifically, the separation of some quantum correlation sets has been achieved \cite{cleve2017perfect, dykema2017non, slofstra2017set, coladangelo2018unconditional, goh2018geometry, musat2018non}. It is this last point that inspired this paper, the hope that exploring how correlations behave under composition would reveal properties that distinguish different correlation sets.

Category theory is the natural level of abstraction to study composition. For this work we examine the categorical notions of one-to-one (section and monomorphism) and onto (retraction and epimorphism) and characterize when a synchronous nonlocal games satisfies each of these. We also get analogous characterizations for the synchronous classical, quantum, and nonsignaling correlation sets. Unfortunately, our hope was not realized in that the property for being a section (or monomomorphism or retraction or epimorphism) is the same for synchronous classical, quantum, or nonsignaling correlations. For this reason we did not define approximating, spacial, or commuting of quantum correlation sets as these properties are too course to distinguish them. Nonetheless, our techniques for producing these characterizations rely on detailed construction of correlations that have particular properties and so we hope these methods are useful for other works.

\end{section}

\begin{section}{Nonlocal games and synchronous correlations}

In our previous work \cite{lackey2017nonlocal} we examined a class of nonlocal games as a generalization functions between finite sets. Given an input set $X$ and an output set $Y$, two players---Alice and Bob---receive $x_A, x_B\in X$ respectively and must then produce $y_A, y_B\in Y$ respectively. Their strategy defines a conditional probability distribution, or \emph{correlation} $p(y_A,y_B \:|\: x_A, x_B)$. We say they win the game if whenever $x_A = x_B$ the produce outputs such that $y_A = y_B$. That is, they win if they play a \emph{synchronous} correlation:
\begin{equation}\label{equation:synchronous}
    p(y_A,y_B \:|\: x, x) = 0 \text{ whenever $x \in X$ and $y_A \not= y_B$ in $Y$.}
\end{equation}

Synchronous correlations generalize functions in that given a function $f:X \to Y$, Alice and Bob choose the strategy of applying $F$ to their respective inputs. This gives the correlation
$$p(y_A, y_B \:|\; x_A, x_B) = \indicator{y_A = f(x_A)}\indicator{y_B = f(x_B)}$$
where we write $\indicator{S}$ for the indicator function of the set $S$ and abbreviation $\{y = f(x)\} = \{ (x,y) \::\: y = f(x)\}$.

Beyond the category of finite sets and functions between them, we study three categories of nonlocal games corresponding to various restrictions we place on the allowed correlation sets. A \emph{hidden variables}, or simply \emph{classical}, correlation takes the form
$$p(y_A, y_B \:|\: x_A, x_B) = \sum_{\omega\in\Omega} \mu(\omega) \Pr_A(y_A \:|\: x_A, \omega) \Pr_B(y_B \:|\: x_B, \omega)$$
for some (discrete) probability space $(\Omega, \mu)$ and ``local'' conditional probabilities $\Pr_A$ and $\Pr_B$. A every synchronous classical correlation can be expressed in the form
\begin{equation}\label{equation:classical:form}
p(y_A, y_B \:|\: x_A, x_B) = \sum_{f\in\Hom(X,Y)} \mu(f) \indicator{y_A = f(x_A)}\indicator{y_B = f(x_B)}
\end{equation}
where we write $\Hom(X,Y)$ for the set of functions from $X$ to $Y$.

A \emph{quantum} correlation takes the form
$$p(y_A, y_B \:|\: x_A, x_B) = \tr(\rho(E^{x_A}_{y_A}\otimes F^{x_B}_{y_B}))$$
where $\rho$ is a density operator on a finite dimensional Hilbert space $\h(H)_A\otimes\h{H}_B$ with $\{E^x_y\}_{y\in Y}$ and $\{F^x_y\}_{y\in Y}$ being positive operator valued measures (POVMs) on $\h{H}_A$ and $\h{H}_B$ respectively for each $x\in X$. For synchronous quantum correlations it can be shown that the $E^x_y$ and $F^x_y$ must be already be projection-valued \cite{cameron2007quantum, mancinska2016quantum}. A general synchronous correlations can be expressed as a convex sum of those where $\rho$ is maximally entangled, and these latter such correlation can be further reduced to tracial states
\begin{equation}\label{equation:quantum:form}
p(y_A, y_B \:|\: x_A, x_B) = \frac{1}{d}\tr(\Pi^{x_A}_{y_A}\Pi^{x_B}_{y_B})
\end{equation}
where $\{\Pi^x_y\}_{y\in Y}$ is a projection-valued measure on a Hilbert space $\h{H}$ and $d = \dim(\h{H})$ \cite{paulsen2016estimating, lackey2017nonlocal}.

A \emph{nonsignaling} correlation \cite{popescu1994quantum} satisfies
\begin{equation}\label{equation:nonsignaling-condition}
\begin{array}{rcl}
\sum_{y_B} p(y_A, y_B \:|\: x_A, x_B) &=& \sum_{y_B} p(y_A, y_B \:|\: x_A, x_B') \text{ for all $y_A, x_A, x_B, x_B'$,}\\
\sum_{y_A} p(y_A, y_B \:|\: x_A, x_B) &=& \sum_{y_A} p(y_A, y_B \:|\: x_A', x_B) \text{ for all $y_B, x_A, x_B, x_A'$.}
\end{array}
\end{equation}
That is, the marginals over Bob's outputs does not depend on Bob's input, and similarly for Alice. In later computation we will find it convenient to express this condition as
$$\sum_{y_B} p(y_A, y_B \:|\: x_A, x_B) = \sum_{y_B} p(y_A, y_B \:|\: x_A, \cdot),$$
with ``$\cdot$'' used to emphasize that the sum does not depend on this argument.

In this work we need three technical lemmas from \cite{lackey2017nonlocal}, which we reproduce here for convenience. In fact, each of these lemmas was created for this work but proved useful to characterize synchronous Bell inequalities.

\begin{lemma}[\cite{lackey2017nonlocal}]\label{lemma:2-point-domain:nonsignaling}
    Let $Y$ be a finite set and $u = u(y_A,y_B)$ and $v = v(y_A,y_B)$ be probability distributions on $Y^2$ such that for all $y \in Y$
    \begin{enumerate}
        \item $\sum_{y'} u(y,y') = \sum_{y'} v(y',y)$ and
        \item $\sum_{y'} u(y',y) = \sum_{y'} v(y,y').$
    \end{enumerate}
    Write $\theta(y)$ and $\phi(y)$ for these two sums respectively and define
    \begin{align*}
        p(y_A,y_B\:|\:0,0) &= \indicator{y_A=y_B}\theta(y_A)\\
        p(y_A,y_B\:|\:0,1) &= u(y_A,y_B)\\
        p(y_A,y_B\:|\:1,0) &= v(y_A,y_B)\\
        p(y_A,y_B\:|\:1,1) &= \indicator{y_A=y_B}\phi(y_A).
    \end{align*}
    Then $p$ is a synchronous nonsignaling correlation. Moreover every nonsignaling correlation with domain $\{0,1\}$ arises this way.
\end{lemma}

\begin{lemma}[\cite{lackey2017nonlocal}]\label{lemma:2-point-domain:classical}
    Let $Y$ be a finite set and $u = u(y_A,y_B)$ be a probability distribution on $Y^2$. Write
    \begin{enumerate}
        \item $\theta(y) = \sum_{y'} u(y,y')$ and
        \item $\phi(y) = \sum_{y'} u(y',y).$
    \end{enumerate}
    Define
    \begin{align*}
        p(y_A,y_B\:|\:0,0) &= \indicator{y_A=y_B}\theta(y_A)\\
        p(y_A,y_B\:|\:0,1) &= u(y_A,y_B)\\
        p(y_A,y_B\:|\:1,0) &= u(y_B,y_A)\\
        p(y_A,y_B\:|\:1,1) &= \indicator{y_A=y_B}\phi(y_A).
    \end{align*}
    Then $p$ is a synchronous classical correlation. Moreover every nonsignaling correlation with domain $\{0,1\}$ arises this way.
\end{lemma}

\begin{lemma}[\cite{lackey2017nonlocal}]\label{lemma:2-point-range:nonsignaling}
    Suppose $|X| \geq 2$ and let $w = w(x_A,x_B)$ be a nonnegative function on $X^2$ such that for every $x_A,x_B \in X$
    \begin{enumerate}
        \item $w(x_A,x_B) \leq w(x_A,x_A)$,
        \item $w(x_A,x_B) \leq w(x_B,x_B)$, and
        \item $w(x_A,x_A) + w(x_B,x_B) \leq 1 + w(x_A,x_B)$.
    \end{enumerate}
    Define
    \begin{align*}
        p(0,0\:|\:x_A,x_B) &= 1 + w(x_A,x_B) - w(x_A,x_A) - w(x_B,x_B)\\
        p(0,1\:|\:x_A,x_B) &= w(x_B,x_B) - w(x_A,x_B)\\
        p(1,0\:|\:x_A,x_B) &= w(x_A,x_A) - w(x_A,x_B)\\
        p(1,1\:|\:x_A,x_B) &= w(x_A,x_B).
    \end{align*}
    Then $p$ is a synchronous nonsignaling correlation from $X$ to $\{0,1\}$. Moreover every synchronous nonsignaling correlation from $X$ to $\{0,1\}$ arises in this way.
\end{lemma}

\end{section}

\begin{section}{Categories of synchronous nonlocal games}

In this work we aim to understand the structure of certain synchronous correlations when treated as homomorphisms in a category. Our prototype is the category \category{FinSet} whose objects are finite sets and morphisms are functions between them. Above we used the notation $\Hom(X,Y)$ for the set of functions between $X$ and $Y$, which is consistent with our use here.

It is currently fashionable to ``categorify'' quantum foundations \cite{abramsky2009categorical, baez2006quantum, chiribella2011informational, coecke2016generalised, tull2016operational}, semantics \cite{coecke2016categories, dixon2009graphical, hines2013types, selinger2010survey}, protocols \cite{abramsky2004categorical, vicary2012higher}, and computation \cite{lee2015computation, barnum2018oracles} as a form of generalized probability theory and resource theory, \cite{chiribella2016operational, coecke2013resource, coecke2016mathematical}. The symmetric monoidal category \cite{maclane1963natural} is the core construction in this approach. While we will not dwell on the monoid product in this work, as it serves as our motivation for our study we simply indicate that in the category \category{FinSet} it is the Cartesian product that acts as a binary operation $X\otimes Y := X \times Y$ 
\begin{enumerate}
    \item being associative $(X\otimes Y)\otimes Z \cong X \otimes (Y\otimes Z)$,
    \item being commutative $X\otimes Y \cong Y \otimes X$, and
    \item having an identity $I$ where $I \otimes X \cong X \cong X \otimes I$.
\end{enumerate}
It is important to note that these monoidal structures are only required to hold up to isomorphism. For example, in \category{FinSet} the identity is any choice of singleton set. Categories where we have equality in properties (1) and (3) are called ``strict'' monoidal categories. For example, the category of Hilbert spaces and isometric maps is strict under the tensor product.

Here we enlarge the category \category{FinSet} by taking $\Hom(X,Y)$ to include not only functions from $X$ to $Y$, but also nonlocal games with input set $X$ and output set $Y$. Specifically, let $\Hom^S(X,Y)$ be the set of nonlocal games from $X$ to $Y$ with synchronous correlations and we define $\Hom^S_{HV}(X,Y)$, $\Hom^S_Q(X,Y)$, and $\Hom^S_{NS}(X,Y)$ to be the subsets of $\Hom^S(X,Y)$ whose correlations are classical, quantum, and nonsignaling respectively. Clearly
$$\Hom(X,Y) \subseteq \Hom^S_{HV}(X,Y) \subseteq \Hom^S_Q(X,Y) \subseteq \Hom^S_{NS}(X,Y) \subseteq \Hom^S(X,Y).$$
In order that these be the sets of morphisms for categories (in each case whose objects are finite sets), we need to define a composition that is closed and associative. Note that the identity function $\id_X\in\Hom(X,X)$ already exists in each of these hom-sets and so can serve as the identity in our enlarged categories.

We take the most natural composition rule: Kolmogorov's Law. Namely, given correlations
$$p(y_A, y_B \:|\: x_A, x_B) \text{ and } p(z_A, z_B \:|\: y_A, y_B)$$
we define the composition as
$$(q\circ p)(z_A, z_B \:|\: x_A, x_B) = \sum_{y_A, y_B} q(z_A, z_B \:|\: y_A, y_B) p(y_A, y_B \:|\: x_A, x_B).$$
Any correlation is associated to a $(|Y|^2\times |X|^2)$-column stochastic matrix defined by $P = [p(y_A, y_B \:|\: x_A, x_B)]$. The Kolmogorov Law above is just matrix multiplication of the associated stochastic matrices and is therefore associative. The challenge is to show that this composition is closed in each of our categories. We do this case by case as follows.

\begin{lemma}
    The composition of synchronous correlations is synchronous.
\end{lemma}
\begin{proof}
    This is a simple computation based on (\ref{equation:synchronous}):
    \begin{align*}
        (q\circ p)(z_A, z_B \:|\: x, x) &= \sum_{y_A, y_B} q(z_A, z_B \:|\: y_A, y_B) p(y_A, y_B \:|\: x, x)\\
        &= \sum_{y} q(z_A, z_B \:|\: y, y) p(y, y \:|\: x, x)\ =\ 0
    \end{align*}
    when $z_A \not= z_B$.
\end{proof}

\begin{proposition}
    The composition of classical correlation is classical.
\end{proposition}
\begin{proof}
    Let $p,q$ be classical correlations with
    \begin{align*}
        p(y_A, y_B \:|\: x_A, x_B) &= \sum_{\omega\in\Omega} \mu(\omega) \Pr^{(p)}_A(y_A \:|\: x_A, \omega) \Pr^{(p)}_B(y_B \:|\: x_B, \omega) \text{ and}\\
        q(z_A, z_B \:|\: y_A, y_B) &= \sum_{\upsilon\in\Upsilon} \nu(\upsilon) \Pr^{(q)}_A(z_A \:|\: y_A, \upsilon) \Pr^{(q)}_B(z_B \:|\: y_B, \upsilon).
    \end{align*}
    Take the probability space $(\Xi,\rho) = (\Omega\times\Upsilon, \mu\times\nu)$ and define Alice's local correlation as
    $$\Pr_A(z_A \:|\: x_A,\omega,\upsilon) = \sum_{y_A} \Pr_A^{(q)}(z_A \:|\: y_A, \upsilon)\Pr^{(p)}_A(y_A \:|\: x_A, \omega),$$
    and similarly for Bob. Then
    \begin{align*}
        (q\circ p)(z_A, z_B \:|\: x_A, x_B) &= \sum_{y_A, y_B} q(z_A, z_B \:|\: y_A, y_B) p(y_A, y_B \:|\: x_A, x_B)\\
        &= \sum_{\omega\in\Omega, \upsilon\in\Upsilon} \mu(\omega)\nu(\upsilon) \sum_{y_A} \Pr_A^{(q)}(z_A \:|\: y_A, \upsilon)\Pr^{(p)}_A(y_A \:|\: x_A, \omega)\\
        & \qquad\qquad\qquad \cdot\ \sum_{y_B}\Pr_B^{(q)}(z_B \:|\: y_B, \upsilon)\Pr^{(p)}_B(y_B \:|\: x_B, \omega)\\
        &= \sum_{(\omega,\upsilon)\in\Xi} \rho(\omega,\nu)\Pr_A(z_A \:|\: x_A,\omega,\upsilon)\Pr_B(z_B \:|\: x_B,\omega,\upsilon).
    \end{align*}
\end{proof}

\begin{proposition}
    The composition of quantum correlations is quantum.
\end{proposition}
\begin{proof}
    Let $p,q$ be quantum correlations with
    \begin{align*}
        p(y_A, y_B \:|\: x_A, x_B) &= \tr(\rho(E^{x_A}_{y_A}\otimes F^{x_B}_{y_B})) \text{ and}\\
        q(z_A, z_B \:|\: y_A, y_B) &= \tr(\sigma(M^{y_A}_{z_A}\otimes N^{y_B}_{z_B})),
    \end{align*}
    where $\rho$ is a density operator on $\h{H}_A \otimes \h{H}_B$ and $\sigma$ on $\h{K}_A \otimes \h{K}_B$, and $\{E^x_y\}$, $\{F^x_y\}$, $\{M^x_y\}$, and $\{N^x_y\}$ are POVMs on the appropriate spaces. Then $\rho\otimes\sigma$ is a density operator on $\h{H}_A \otimes \h{H}_B \otimes \h{K}_A \otimes \h{K}_B$ however to align the Hilbert spaces we want the density $R = S^{-1}(\rho\times\sigma)S$ where $S:\h{H}_A \otimes \h{K}_A \otimes \h{H}_B \otimes \h{K}_B \to \h{H}_A \otimes \h{H}_B \otimes \h{K}_A \otimes \h{K}_B$ is the natural swap transformation.
    
    Next consider the family of operators $\{\bar{E}^x_z\}_{z\in Z}$ (for each $x\in X$) given by
    $$\bar{E}^x_z = \sum_{y\in Y} E^x_y \otimes M^y_z.$$
    These are nonnegative operators on $\h{H}_A\otimes\h{K}_A$; we show these form POVMs by computing:
    \begin{align*}
        \sum_{z\in Z} \bar{E}^x_z &= \sum_{y \in Y} E^x_y \otimes \sum_{z\in Z} M^y_z\\
        &= \sum_{y \in Y} E^x_y \otimes \openone_{\h{K}_A}\\
        &= \openone_{\h{H}_A}\otimes \openone_{\h{K}_A}.
    \end{align*}
    Identically, $\bar{F}^x_z = \sum_{y\in Y} F^x_y\otimes N^y_z$ define POVMs on $\h{H}_B \otimes \h{K}_B$.
    
    Finally we compute
    \begin{align*}
        \tr(R(\bar{E}^{x_A}_{z_A}\times \bar{F}^{x_B}_{z_B})) &= \tr(S^{-1}(\rho\otimes\sigma)S(\bar{E}^{x_A}_{z_A}\times \bar{F}^{x_B}_{z_B}))\\
        &= \sum_{y_A, y_B} \tr((\rho\otimes\sigma)S(E^{x_A}_{y_A} \otimes F^{x_B}_{y_B} \otimes M^{y_A}_{z_A} \times N^{y_B}_{z_B})S^{-1})\\
        &= \sum_{y_A, y_B} \tr((\rho\otimes\sigma)(E^{x_A}_{y_A} \otimes M^{y_A}_{z_A} \otimes F^{x_B}_{y_B} \otimes N^{y_B}_{z_B}))\\
        &= \sum_{y_A, y_B} \tr(\rho(E^{x_A}_{y_A}\otimes F^{x_B}_{y_B}))\tr(\sigma(M^{y_A}_{z_A}\otimes N^{y_B}_{z_B}))\\
        &= (q\circ p)(z_A, z_B \:|\: x_A, x_B).
    \end{align*}
\end{proof}

\begin{proposition}
    The composition of nonsignaling correlations is nonsignaling.
\end{proposition}
\begin{proof}
    We simply compute
    \begin{align*}
        \sum_{z_B} (q\circ p)(z_A, z_B \:|\: x_A, x_B) &= \sum_{y_A, y_B, z_B} q(z_A, z_B \:|\: y_A, y_B) p(y_A, y_B \:|\: x_A, x_B)\\
        &= \sum_{y_A, z_B} q(z_A, z_B \:|\: y_A, \cdot) \sum_{y_B} p(y_A, y_B \:|\: x_A, x_B)\\
        &= \sum_{y_A, z_B} q(z_A, z_B \:|\: y_A, \cdot) \sum_{y_B} p(y_A, y_B \:|\: x_A, \cdot),
    \end{align*}
    where as above we use ``$\cdot$'' to indicate the formula does not depend on this argument. This is the first of the nonsignaling conditions (\ref{equation:nonsignaling-condition}). The other follows identically.
\end{proof}

Write $\category{FinSet}^S_{HV}$, $\category{FinSet}^S_Q$, $\category{FinSet}^S_{NS}$, and $\category{FinSet}^S$ for the categories whose objects are the finite sets and morphisms are $\Hom^S_{HV}(X,Y)$, $\Hom^S_Q(X,Y)$, $\Hom^S_{NS}(X,Y)$, and $\Hom^S(X,Y)$ respectively.

\end{section}

\begin{section}{Injectivity}

In general categories there are two notions of injective homomorphisms: sections (left invertible) and monomorphisms (left cancelable). Every section is a monomorphism in generality; in $\category{FinSet}$ both of these notions are equivalent to a function being one-to-one. Recall that in a general category $\category{C}$ a morphism $f\in\Hom_\category{C}(A,B)$ is a \emph{section} if there exists a $g \in \Hom_\category{C}(B,A)$ with $g\circ f = \mathrm{id}_A$. We say $f$ is a \emph{monomorphism} if whenever $g,h\in \Hom_\category{C}(Z,A)$ satisfy $f\circ g = f \circ h$ then $g = h$. A natural question is to whether there are any other functions that can be inverted, if the inverse is only required to be a synchronous nonlocal game. The answer is no as proven in the next result.

\begin{lemma}\label{lemma:synchronous:section}
    Suppose $p\in\Hom(X,Y)$ and $q\in\Hom^S(Y,X)$ with $q\circ p = \id_X$. Then $p$ corresponds to a one-to-one function.
\end{lemma}
\begin{proof}
    Let $p,q$ be as stated and write $f:X\to Y$ for the function association to $p$. Writing out $q\circ p = \id_X$ we have for all $(x_A,x_B)\in X^2$ that
    \begin{align*}
        \indicator{x_A = x_A'}\indicator{x_B = x_B'} &= \sum_{y_A, y_B} q(x_A', x_B' \:|\: y_A, y_B) p(y_A, y_B \:|\: x_A, x_B) \\
        &= \sum_{y_A, y_B} q(x_A', x_B' \:|\: y_A, y_B) \indicator{y_A = f(x_A)}\indicator{y_B = f(x_B)} \\
        &= q(x_A', x_B' \:|\: f(x_A), f(x_B)).
    \end{align*}
    Now suppose $f(x_0) = f(x_1)$. Inserting this into the above we get
    $$1 = q(x_0,x_1 \:|\: f(x_0), f(x_1)).$$
    However $q$ is synchronous and hence $x_0 = x_1$.
\end{proof}

\begin{theorem}\label{theorem:section:synchronous}
    In $\category{FinSet}^S$ the sections are precisely the deterministic correlations
    $$p(y_A, y_B|x_A, x_B) = \indicator{y_A = f_A(x_A,x_B)}\indicator{y_B = f_B(x_A,x_B)}$$
    such that $F = (f_A,f_B):X^2 \to Y^2$ is one-to-one and $f_A(x_A,x_B) = f_B(x_A,x_B)$ if and only if $x_A = x_B$.
\end{theorem}
\begin{proof}
    Let $p \in \Hom^S(X,Y)$ be a section and write $q \in \Hom^S(Y,X)$ for a correlation with $q\circ p = \id_X$. Fix $(x_A,x_B) \in X^2$ and and write $\mu(y_A,y_B) = p(y_A, y_B|x_A, x_B)$. Now,
    \begin{align}\nonumber
        \indicator{x_A' = x_A}\indicator{x_B' = x_B} &= (q\circ p) (x_A', x_B' \:|\: x_A, x_B)\\\label{equation:section:synchronous:proof1}
        &= \sum_{y_A,y_B} q(x_A', x_B' \:|\: y_A, y_B) p(y_A, y_B \:|\: x_A, x_B)\\
        &= \sum_{y_A,y_B} q(x_A', x_B' \:|\: y_A, y_B) \mu(y_A,y_B).\nonumber
    \end{align}
    As $q$ is a stochastic map, and entropy $H$ being a concave function, we have $H(\indicator{x_A' = x_A}\indicator{x_B' = x_B}) \geq H(\mu)$. However $H(\indicator{x_A' = x_A}\indicator{x_B' = x_B}) = 0$ and hence $\mu$ is a point mass concentrated as some point $(\hat{y}_A, \hat{y}_B)$. We define the function $F:X^2 \to Y^2$ by $F(x_A, x_B) = (\hat{y}_A, \hat{y}_B)$ obtained this way. Thus for $F = (f_A,f_B)$ we have
    $$p(y_A, y_B|x_A, x_B) = \indicator{y_A = f_A(x_A,x_B)}\indicator{y_B = f_B(x_A,x_B)}.$$
    
    Inserting this formula into (\ref{equation:section:synchronous:proof1}) produces
    \begin{equation}\label{equation:section:synchronous:proof2}
    \indicator{x_A' = x_A}\indicator{x_B' = x_B} = q(x_A', x_B' \:|\: f_A(x_A,x_B), f_B(x_A,x_B)).
    \end{equation}
    Therefore $q$ restricted to the image of $F$ in $Y^2$ is just a function back into $X^2$. Extend this function in any way to a function $G:Y^2 \to X^2$, which then satisfies $G\circ F = \id_{X^2}$ and hence $F$ is one-to-one.
    
    As $p$ is synchronous, from equation (\ref{equation:synchronous}) we get
    $$\indicator{y_A = y_B} = \indicator{y_A = f_A(x,x)}\indicator{y_B = f_B(x,x)},$$
    and so $f_A(x,x) = f_B(x,x)$ for all $x\in X$. Conversely, if $f_A(x_A,x_B) = f_B(x_A,x_B)$
    then as $q$ is synchronous we have from (\ref{equation:section:synchronous:proof2}) that
    $$\indicator{x_A' = x_A}\indicator{x_B' = x_B} = \indicator{x_A' = x_B'}.$$
    That is, we must have $x_A = x_B$, as desired.
    
    Conversely, let $p$ be given as stated in the theorem. Clearly $p$ is synchronous given its formula. Now fix any $\hat{x}\in X$ and define
    $$G(y_A,y_B) = \left\{\begin{array}{cl} 
    (x_A,x_B) & \text{if $(y_A,y_B) = F(x_A,x_B)$,}\\
    (\hat{x},\hat{x}) & \text{otherwise.}\end{array}\right.$$
    This is well defined since $F$ is one-to-one. Write $G = (g_A, g_B)$ and define
    $$q(x_A,x_B \:|\: y_A,y_B) = \indicator{x_A = g_A(y_A, y_B)} \indicator{x_B = g_B(y_A, y_B)}.$$
    Note that $(y,y) = F(x_A, x_B)$ if and only if $x_A = x_B$, and $(y,y)$ not in the image of $F$ we have $G(y,y) = (\hat{x},\hat{x})$. Hence
    $$q(x_A,x_B \:|\: y,y) = 0 \text{ whenever $x_A \not= x_B$,}$$
    and so $q \in \Hom^S(Y,X)$. Finally,
    \begin{align*}
        (q\circ p)(x_A', x_B' \:|\: x_A, x_B) &= \sum_{y_A,y_B} \indicator{x_A' = g_A(y_A, y_B)} \indicator{x_B' = g_B(y_A, y_B)} \indicator{y_A = f_A(x_A,x_B)}\indicator{y_B = f_B(x_A,x_B)}\\
        &= \indicator{x_A' = g_A(f_A(x_A,x_B), f_B(x_A,x_B))} \indicator{x_B' = g_B(f_A(x_A,x_B), f_B(x_A,x_B))}\\
        &= \indicator{x_A' = x_A}\indicator{x_B' = x_B}
    \end{align*}
    and so $p$ is a section.
\end{proof}

\begin{corollary}\label{corollary:section:nonsignaling}
    The sections in $\category{FinSet}^S_{NS}$ are precisely the sections in $\category{FinSet}$.
\end{corollary}
\begin{proof}
    Since $\Hom(X,Y)\subseteq \Hom_{NS}^S(X,Y)$ a section in $\Hom(X,Y)$ is also a section in $\Hom_{NS}^S(X,Y)$.

    Conversely, let $p\in\Hom_{NS}^S(X,Y)$ be a section. Then it is also a section of $\Hom^S(X,Y)$ and so from the theorem we have
    \begin{equation}\label{equation:section:nonsignaling:proof1}
        p(y_A, y_B|x_A, x_B) = \indicator{y_A = f_A(x_A,x_B)}\indicator{y_B = f_B(x_A,x_B)}
    \end{equation}
    such that $F = (f_A,f_B):X^2 \to Y^2$ is one-to-one and $f_A(x_A,x_B) = f_B(x_A,x_B)$ if and only if $x_A = x_B$. Since $p$ is nonsignaling
    $$\sum_{y_B} p(y_A, y_B \:|\: x_A, x_B) = \indicator{y_A = f_A(x_A,x_B)}$$
    is independent of $x_B$. That is $f_A(x_A,x_B) = f(x_A)$ alone. Identically, $f_B(x_A,x_B) = g(x_B)$. But then $f(x) = f_A(x,x) = f_B(x,x) = g(x)$ for all $x\in X$ and therefore
    (\ref{equation:section:nonsignaling:proof1}) simplifies to
    $$p(y_A, y_B|x_A, x_B) = \indicator{y_A = f(x_A)}\indicator{y_B = f(x_B)}.$$
    That is $p\in\Hom(X,Y)$.
\end{proof}

\begin{corollary}\label{corollary:section:other}
    The sections in $\category{FinSet}^S_{HV}$ and $\category{FinSet}^S_{HV}$ are precisely the sections in $\category{FinSet}$.
\end{corollary}
\begin{proof}
    Since $\Hom(X,Y)\subseteq \Hom^S_{HV}(X,Y) \subseteq \Hom^S_Q(X,Y)$ the one-to-one functions define sections in each of these categories. Moreover a section in each of these categories must form a section in $\category{FinSet}^S_{NS}$, and therefore must be in $\Hom(X,Y)$.
\end{proof}

The above results show that the sections and the one-to-one functions are the same sets (having identified a function with a deterministic synchronous correlation). However being monic is a strictly weaker condition in general as the following results will show.

\begin{lemma}\label{lemma:monomorphism:nullspace}
    In each of $\category{FinSet}^S_{HV}$, $\category{FinSet}^S_Q$, $\category{FinSet}^S_{NS}$, and $\category{FinSet}^S$ any correlation whose stochastic matrix has zero right nullspace is a monomorphism.
\end{lemma}
\begin{proof}
    Let $p\in\Hom^S(X,Y)$ have stochastic map $P$ with zero right nullspace. If $q_1, q_2 \in\Hom^S(U,X)$ are correlations with $p\circ q_1 = p\circ q_2$ then their stochastic maps $Q_1,Q_2$ satisfy $P(Q_1 - Q_2) = 0$. In particular each column of $Q_1-Q_2$ is in the null space of $P$ and therefore $Q_1 = Q_2$ and $q_1 = q_2$. The identical argument works in the other categories.
\end{proof}

The remainder of this section is dedicated to proving that ``zero right nullspace'' characterizes when a synchronous correlation in any of these categories is a monomorphism. Let $P$ be a stochastic map and $P\vec{u} = \vec{0}$. We can decompose $\vec{u} = \vec{u}_+ - \vec{u}_-$ where each of $\vec{u}_\pm$ has nonnegative entries. Notice that
\begin{align*}
    0 &= (1, \dots, 1)P\vec{u}\\
    &= (1, \dots, 1)(\vec{u}_+ - \vec{u}_-)\\
    &= \|\vec{u}_+\|_1 - \|\vec{u}_-\|_1.
\end{align*}
Assuming $\vec{u}\not=\vec{0}$---which is what we ultimately aim to contradict---we may renormalize $\vec{u}$ if necessary so that $\|\vec{u}_\pm\|_1 = 1$, and hence each of $\vec{u}_\pm$ defines a probability distribution. Our program is then to use $\vec{u}_\pm$ to construct appropriate correlations $q_\pm$ from which we can deduce $\vec{u} = \vec{0}$.

\begin{proposition}\label{proposition:synchronous:monomorphism}
  In $\category{FinSet}^S$ the monomorphisms are precisely those correlations whose stochastic matrices have zero right nullspace.
\end{proposition}
\begin{proof}
    One implication is provided by Lemma \ref{lemma:monomorphism:nullspace}. For the converse, suppose $p\in\Hom^S(X,Y)$ is a monomorphism with associated stochastic matrix $P$, and suppose $P\vec{v} = \vec{0}$ for some $\vec{v}\in\mathbb{R}^{|X|^2}$. As above, presume $\vec{v} \not= \vec{0}$ and decompose $\vec{v} = \vec{v}_+ - \vec{v}_-$ where each $\vec{v}_\pm$ is nonnegative with sum equal one. 
    
    From $\vec{v}_\pm$ we construct $q_\pm\in\Hom^S(\{0,1\},X)$ as follows: fix any $\hat{x} \in X$ at define
    \begin{align*}
        q_\pm(x_A,x_B|0,0) &= \left\{\begin{array}{cl} 1 & \text{if $x_A=x_B=\hat{x}$}\\ 0 & \text{otherwise,}\end{array}\right.\\
        q_\pm(x_A,x_B|0,1) &= v_\pm(x_A,x_B)\\
        q_\pm(x_A,x_B|1,0) &= \left\{\begin{array}{cl} 1 & \text{if $x_A=x_B=\hat{x}$}\\ 0 & \text{otherwise,}\end{array}\right.\\
        q_\pm(x_A,x_B|1,1) &= \left\{\begin{array}{cl} 1 & \text{if $x_A=x_B=\hat{x}$}\\ 0 & \text{otherwise.}\end{array}\right.\\
    \end{align*}
    Since $q_\pm(x_A,x_B|0,0) = q_\pm(x_A,x_B|1,1) = 0$ when $x_A\not= x_B$, both $q_+$ and $q_-$ are synchronous. Their associated stochastic maps satisfy $Q_+ - Q_- = \left( \vec{0}, \vec{v}, \vec{0}, \vec{0}\right)$, and hence $P(Q_+ - Q_-) = 0$. This implies that $p\circ q_+ = p\circ q_-$, from which we conclude $q_+ = q_-$ as $p$ is a monomorphism. Therefore $Q_+ = Q_-$ and so we must have had $\vec{v} = \vec{0}$ as desired.
\end{proof}

Unfortunately this proof does not carry over to the other categories easily. The construction of $q_\pm$ is specialized to general synchronous correlations and would not define homomorphisms in the more restrictive categories. Namely to construct $q_1,q_2$ in the proper category leads to stochastic maps $Q_1 - Q_2$ which cannot not simply be formed by the given vector $\vec{v}$ in the kernel. Consequently we will need to know more information about the structure of the kernel of a stochastic map in each of the categories.

\begin{lemma}\label{lemma:monomorphism:nonsignaling}
    Let $p$ be a nonsignaling synchronous correlation, $P$ its stochastic matrix, and suppose $P\vec{u} = \vec{0}$. Define $w_1(x_A,x_B) = \indicator{x_A=x_B}\sum_z u(x_A,z)$ and $w_2(x_A,x_B) = \indicator{x_A=x_B}\sum_z u(z,x_B)$. Then we have $P\vec{w}_1 = P\vec{w}_2 = \vec{0}$.
\end{lemma}
\begin{proof}
  Writing out $P\vec{u} = \vec{0}$ produces
  $$0 = \sum_{x_A,x_B} p(y_A,y_B\:|\: x_A,x_B) u(x_A,x_B) \text{ for all $y_A,y_B\in Y$.}$$
  Taking the sum on $y_B$ and using that $p$ is nonsignaling gives
  \begin{align*}
      0 &= \sum_{x_A,x_B,y_B} p(y_A,y_B\:|\: x_A,x_B)u(x_A,x_B)\\
      &= \sum_{x_A,y_B} p(y_A,y_B\:|\: x_A,\cdot)\sum_{x_B} u(x_A,x_B),
  \end{align*}
  where we continue to use ``$\cdot$'' to highlight that the correlation does not depend on $x_B$. In fact, since it does not depend on $x_B$, we could evaluate that argument at $x_A$ and use the synchronicity of $p$ to find
  $$\sum_{y_B} p(y_A,y_B\:|\: x_A,\cdot) = \sum_{y_B} p(y_A,y_B\:|\: x_A,x_A) = p(y_A,y_A\:|\: x_A,x_A).$$
  Inserting this into the above equation shows
  \begin{align*}
      0 &= \sum_{x_A} p(y_A,y_A\:|\: x_A,x_A)\sum_z u(x_A,z)\\
      &= \sum_{x_A,x_B} p(y_A,y_A\:|\: x_A,x_B)\indicator{x_A = x_B}\sum_z u(x_A,z)\\
      &= \sum_{x_A,x_B} p(y_A,y_A\:|\: x_A,x_A)w_1(x_A,x_B).
  \end{align*}
  Finally, when $y_A \not= y_B$ we have
  $$\sum_{x_A,x_B} p(y_A,y_B\:|\: x_A,x_A)w(x_A,x_B) = \sum_{x_A} p(y_A,y_B\:|\: x_A,x_A)\sum_z u(x_A,z)
  = 0$$
  as $p$ is synchronous. Therefore $\sum_{x_A,x_B} p(y_A,y_B\:|\: x_A,x_A)w_1(x_A,x_B) = 0$ for all $y_A,y_B\in Y$. The proof $P\vec{w}_2 = \vec{0}$ starting from the sum on $y_A$ is identical.
\end{proof}

We would like to use Lemma \ref{lemma:2-point-domain:nonsignaling} to generate nonsignaling correlations, with $u = v$, and Lemma \ref{lemma:monomorphism:nonsignaling} almost allows this. The missing piece is the side-condition $\sum_{x'} u(x',x) = \sum_{x'} u(x,x')$, using the notation of the latter lemma, which does not appear to hold in general. We say a correlation is symmetric if it satisfies
$$p(y_A,y_B \:|\: x_A,x_B) = p(y_B,y_A \:|\: x_B,x_A).$$
For symmetric correlations this condition is trivial. Moreover from equations (\ref{equation:classical:form},\ref{equation:quantum:form}) every synchronous classical and quantum correlation is symmetric, and so we can complete our characterization of monomorphisms in these two categories as follows.

\begin{lemma}\label{lemma:monomorphism:symmetric}
    Let $p\in\Hom^S(X,Y)$ be a symmetric with associated stochastic matrix $P$. Suppose $P\vec{u} = \vec{0}$. Then $v(x_A,x_B) = u(x_B,x_A)$ also has $P\vec{v} = \vec{0}$.
\end{lemma}
\begin{proof}
    Simply,
    \begin{align*}
        \sum_{x_A,x_B} p(y_A,y_B\:|\: x_A, x_B) v(x_A,x_B) &= \sum_{x_A,x_B} p(y_A,y_B\:|\: x_A, x_B) u(x_B, x_A)\\
        &= \sum_{x_A,x_B} p(y_A,y_B\:|\: x_B, x_A) u(x_A, x_B)\\
        &= \sum_{x_A,x_B} p(y_B,y_A\:|\: x_A, x_B) u(x_A, x_B)\ =\ 0 \text{ for all $y_A,y_B$.}
    \end{align*}
\end{proof}

\begin{theorem}\label{theorem:monomorphism:symmetric}
    In $\category{FinSet}^S_{HV}$ and $\category{FinSet}^S_Q$ the monomorphisms are precisely those correlations whose stochastic matrices have zero right nullspace.
\end{theorem}
\begin{proof}
    Let $p$ be a monomorphic synchronous either hidden variables or quantum correlation with associated stochastic matrix $P$. In either case $p$ is symmetric and nonsignaling. Suppose $P\vec{u} = \vec{0}$ with $\vec{u}\not=\vec{0}$ and as above decompose $\vec{u} = \vec{u}_+ - \vec{u}_-$ where each of $\vec{u}_\pm$ defines a probability distribution. Define
    \begin{align*}
        v_\pm(x_A,x_B) &= u_\pm(x_B,x_A)\\
        w_\pm(x_A,x_B) &= \openone_{x_A=x_B}\sum_z u_\pm(x_A,z)\\
        z_\pm(x_A,x_B) &= \openone_{x_A=x_B}\sum_z v_\pm(x_A,z)\ =\ \openone_{x_A=x_B}\sum_z u_\pm(z,x_B).
    \end{align*}
    Then the vectors $\vec{v} = \vec{v}_+ - \vec{v}_-$, $\vec{w} = \vec{w}_+ - \vec{w}_-$, and $\vec{z} = \vec{z}_+ - \vec{z}_-$ are in the right nullspace of $P$ from Lemmas \ref{lemma:monomorphism:nonsignaling} and \ref{lemma:monomorphism:symmetric}.
    
    Set
    \begin{align*}
        q_\pm(x_A,x_B\:|\:0,0) &= w_\pm(x_A,x_B)\\
        q_\pm(x_A,x_B\:|\:0,1) &= u_\pm(x_A,x_B)\\
        q_\pm(x_A,x_B\:|\:1,0) &= v_\pm(x_A,x_B)\\
        q_\pm(x_A,x_B\:|\:1,1) &= z_\pm(x_A,x_B).
    \end{align*}
    Then by Lemma \ref{lemma:2-point-domain:classical} both $q_\pm\in\Hom^S_{HV}(\{0,1\},X)$. Moreover 
    $$P(Q_+ - Q_-) = P\left(\begin{array}{cccc} \vec{w}&\vec{u}&\vec{v}&\vec{z}\end{array}\right) = 0$$
    and so $p\circ q_+ = p\circ q_-$. As $p$ is monomomorphism $q_+ = q_-$ and in particular, $\vec{u} = \vec{u}_+ - \vec{u}_- = \vec{0}$ obtaining our contradiction.
\end{proof}

A general nonsignaling correlation will not be symmetric and so the above proof will not work. Fortunately a simple modification can be made that provides the same result.

\begin{theorem}
    In $\category{FinSet}^S_{NS}$ the monomorphisms are precisely those correlations whose stochastic matrices have zero right nullspace.
\end{theorem}
\begin{proof}
    Let $p$ be a monomorphic synchronous nonsignaling correlation with associated stochastic matrix $P$. As before, suppose $P\vec{u} = \vec{0}$ and decompose $\vec{u} = \vec{u}_+ - \vec{u}_-$ into vectors that define probability distributions. Now define
    \begin{align*}
        w_\pm(x_A,x_B) &= \indicator{x_A = x_B} \sum_z u_\pm(x_A,z)\\
        z_\pm(x_A,x_B) &= \indicator{x_A = x_B} \sum_z u_\pm(z,x_B)\\
        v_\pm(x_A,x_B) &= w_\pm(x_A,x_B) + z_\pm(x_A,x_B) - u_\pm(x_A,x_B).
    \end{align*}
    Both $\vec{w} = \vec{w}_+ - \vec{w}_-$ and $\vec{z} = \vec{z}_+ - \vec{z}_-$ are in the nullspace of $P$ by Lemma \ref{lemma:monomorphism:nonsignaling}. But then $\vec{v}$ is also in the nullspace of $P$ as it is a linear combination of $\vec{w}$, $\vec{z}$, and $\vec{u}$. As in the previous proof form the correlations
    \begin{align*}
        q_\pm(x_A,x_B\:|\:0,0) &= w_\pm(x_A,x_B)\\
        q_\pm(x_A,x_B\:|\:0,1) &= u_\pm(x_A,x_B)\\
        q_\pm(x_A,x_B\:|\:1,0) &= v_\pm(x_A,x_B)\\
        q_\pm(x_A,x_B\:|\:1,1) &= z_\pm(x_A,x_B).
    \end{align*}
    If we can verify $\sum_{x'} u_\pm(x,x') = \sum_{x'} v_\pm(x',x)$ and $\sum_{x'}u_\pm(x',x) = \sum_{x'} v_\pm(x,x')$ then by Lemma \ref{lemma:2-point-domain:nonsignaling} both $q_\pm \in \Hom^S_{NS}(\{0,1\},X)$, and the proof from the previous proposition carries through. To verify these two facts, we simply work through their definitions:
    \begin{align*}
        \sum_{x'} v_\pm(x',x) &= \sum_{x'} \left(w_\pm(x',x) + z_\pm(x',x) - u_\pm(x',x)\right)\\
        &= \sum_y u_\pm(x,y) + \sum_y u_\pm(y,x) - \sum_{x'} u_\pm(x',x)\\
        &= \sum_y u_\pm(x,y).
    \end{align*}
    and identically $\sum_{x'} u_\pm(x',x) = \sum_{x'} v_\pm(x,x')$.
\end{proof}

\end{section}

\begin{section}{Surjectivity}

Like injectivity, there are two notions of surjective homomorphisms in a general category: retracts (right invertible) and epimorphisms (right cancelable).  In a general $\category{C}$ a morphism $f\in\Hom_\category{C}(A,B)$ is a \emph{retract} if there exists a $g \in \Hom_\category{C}(B,A)$ with $f\circ g = \mathrm{id}_B$, and a \emph{epimorphism} if whenever $g,h\in \Hom_\category{C}(B,D)$ satisfy $g\circ f = h \circ f$ then $g = h$. As expected, every retract is an epimorphism, and in $\category{FinSet}$ both of these notions are equivalent to that of an onto function. Analogously to sections, retractions correspond to onto functions in categories of nonlocal games.

\begin{theorem}\label{theorem:retraction:synchronous}
    In $\category{FinSet}^S$ the retractions are precisely the deterministic correlations
    $$p(y_A, y_B|x_A, x_B) = \indicator{y_A = f_A(x_A,x_B)}\indicator{y_B = f_B(x_A,x_B)}$$
    such that (i) $F = (f_A,f_B):X^2 \to Y^2$ is onto, (ii) $f_A(x,x) = f_B(x,x)$, and (iii) for each $y\in Y$ we have $f_A(x,x) = y = f_B(x,x)$ for some $x\in X$.
\end{theorem}
\begin{proof}
    If $p\in\Hom^S(X,Y)$ is a retract then $p\circ q = \id_Y$, which is precisely the condition for $q\in\Hom^S(Y,X)$ being a section. By Theorem \ref{theorem:section:synchronous} we have
    $$q(x_A, x_B \:|\: y_A, y_B) = \indicator{x_A = g_A(y_A, y_B)}\indicator{x_B = g_B(y_A, y_B)}$$
    where $G = (g_A,g_B):Y^2 \to X^2$ is one-to-one and $g_A(y_A,y_B) = g_B(y_A,y_B)$ if and only if $y_A = y_B$. Just as in the proof of that theorem fix a $\hat{y}\in Y$ and define
    $$F(x_A, x_B) = \left\{\begin{array}{cl}
        (y_A, y_B) & \text{if $(x_A, x_B) = G(y_A, y_B)$,}\\
        (\hat{y},\hat{y}) & \text{otherwise.}\end{array}\right.$$
    Clearly $F$ is onto since $G$ is defined on all of $Y^2$. If $(x,x) = G(y_A, y_B)$ then we must have $f_A(x,x) = y_A = y_B = f_B(x,x)$; otherwise $f_A(x,x) = \hat{y} = f_B(x,x)$. Finally, fix an arbitrary $y\in Y$. Then we have $x_A = g_A(y,y) = g_B(y,y) = x_B$, and so set $x = x_A = x_B$. By construction this $x$ has $F(x,x) = (y,y)$ or equivalently $f_A(x,x) = y = f_B(x,x)$.
    
    Conversely, suppose $p$ is as given in the statement of the theorem. Clearly $p$ is synchronous given its formula. Construct $G:Y^2 \to X^2$ by selecting some $(x_A,x_B) \in F^{-1}(y_A,y_B)$ for each $(y_A,y_B) \in Y^2$, under the condition that for $(y_A,y_B) = (y,y)$ we select an $(x,x) \in F^{-1}(y,y)$ whose existence is guaranteed. Then set $G(y_A,y_B) = (x_A, x_B)$, and define
    $$q(x_A, x_B \:|\: y_A, y_B) = \indicator{x_A = g_A(y_A, y_B)}\indicator{x_B = g_B(y_A, y_B)}.$$
    As $G(y,y) = (x,x)$ we have $q$ is synchronous and just as in Theorem \ref{theorem:section:synchronous} one computes $p\circ q = \id_Y$. Thus $p\in\Hom^S(X,Y)$ is a retract.
\end{proof}

\begin{corollary}\label{corollary:retraction:nonsignaling}
    The retractions in $\category{FinSet}^S_{NS}$ are precisely the retractions in $\category{FinSet}$.
\end{corollary}

\begin{corollary}\label{corollary:retraction:other}
    The retractions in $\category{FinSet}^S_{HV}$ and $\category{FinSet}^S_{HV}$ are precisely the retractions in $\category{FinSet}$.
\end{corollary}

The proofs of Corollaries \ref{corollary:retraction:nonsignaling} and \ref{corollary:retraction:other} are essentially identical to those of Corollaries \ref{corollary:section:nonsignaling} and \ref{corollary:section:other} respectively and so are omitted. Similarly, we can draw some immediate conclusions about the nature of epimorphism by dualizing the corresponding results on monomorphisms. For example the proof of the following lemma is structurally identical as that of Lemma \ref{lemma:monomorphism:nullspace} and so is left to the reader.

\begin{lemma}\label{lemma:epimorphism:nullspace}
    In each of $\category{FinSet}^S_{HV}$, $\category{FinSet}^S_Q$, $\category{FinSet}^S_{NS}$, and $\category{FinSet}^S$ any correlation whose stochastic map has zero left nullspace is an epimorphism.
\end{lemma}

\begin{theorem}\label{theorem:epimorphism:synchronous}
    In $\category{FinSet}^S$ epimorphisms are precisely those correlations whose stochastic matrices have zero left nullspace.
\end{theorem}
\begin{proof}
    Lemma 21 provides one implication. To prove the converse, suppose $p\in\Hom^S(X,Y)$ is an epimorphism with associated stochastic matrix $P$ and let $\covec{w}P = \covec{0}$. Similar to what we have done before, decompose $\covec{w} = \covec{w}_+ - \covec{w}_-$ into is positive and negative parts. Renormalize if necessary so that $\|\covec{w}_\pm\|_\infty \leq 1$. Define
    \begin{align*}
        q_\pm(0,0\:|\: y_A, y_B) &= 1 - w_\pm(y_A, y_B)\\
        q_\pm(0,1\:|\: y_A, y_B) &= 0\\
        q_\pm(1,0\:|\: y_A, y_B) &= 0\\
        q_\pm(1,1\:|\: y_A, y_B) &= w_\pm(y_A, y_B).
    \end{align*}
    Clearly $q_\pm\in\Hom^S(Y, \{0,1\})$. If $Q_\pm$ are the associated stochastic map then
    $$Q_+ - Q_- = \left(\begin{array}{c} -\covec{w} \\ \covec{0} \\ \covec{0} \\ \covec{w} \end{array}\right),$$
    and therefore $(Q_+ - Q_-)P = 0$ and $q_+\circ p = q_-\circ p$. As $p$ is an epimorphisms we have $Q_+ = Q_-$ and hence $\covec{w} = \covec{0}$.
\end{proof}

Proving converses of Lemma \ref{lemma:monomorphism:nullspace} required intricate knowledge about how synchronous correlations behave in each of our categories. While none of those constructions dualize readily we can obtain the analogous results through different arguments.

\begin{theorem}\label{theorem:epimorphism:nonsignaling}
    In $\category{FinSet}^S_{NS}$ epimorphisms are precisely those correlations whose stochastic matrices have zero left nullspace.
\end{theorem}
\begin{proof}
    Lemma \ref{lemma:epimorphism:nullspace} provides one implication. To prove the converse, suppose $p\in\Hom^S_{NS}(X,Y)$ is an epimorphism with associated stochastic matrix $P$ and let $\covec{w}P = \covec{0}$. Synchronicity gives
    \begin{equation*}
        0 = \sum_{y_A,y_B} w(y_A,y_B) p(y_A,y_B\:|\:x,x) =  \sum_{y} w(y,y) p(y,y\:|\:x,x)
    \end{equation*}
    Define $u(y_A,y_B) = w(y_A,y_A)$ (independent of $y_B$). Then from the nonsignaling condition
    \begin{equation*}
        \sum_{y_A,y_B} u(y_A,y_B) p(y_A,y_B\:|\:x_A,x_B) = \sum_{y_A} w(y_A,y_A) \sum_{y_B} p(y_A,y_B\:|\:x_A,x_B)
    \end{equation*}
    is independent of $x_B$. And so we may take $x_B = x_A$, which gives
    \begin{align*}
        \sum_{y_A,y_B} u(y_A,y_B) p(y_A,y_B\:|\:x_A,x_B) &= \sum_{y_A} w(y_A,y_A) \sum_{y_B} p(y_A,y_B\:|\:x_A,x_A)\\
        &= \sum_{y_A} w(y_A,y_A) p(y_A,y_A\:|\:x_A,x_A) = 0.
    \end{align*}
    That is, $\covec{u}$ also satisfies $\covec{u}P = \covec{0}$. Identically we define $v(y_A,y_B) = w(y_B,y_B)$ which has $\covec{v}P = \covec{0}$.
    
    Scaling $\covec{w}$ if necessary, we may assume $\|\covec{w}\|_\infty \leq \frac{1}{4}$. As before, we decompose $\covec{w} = \covec{w}_1 - \covec{w}_2$, but here we choose
    \begin{align*}
        w_+(y_A,y_B) &= \frac{1}{4}\delta_{y_A,y_B} + \max\{w(y_A,y_B), 0\}\\
        w_-(y_A,y_B) &= \frac{1}{4}\delta_{y_A,y_B} + \max\{-w(y_A,y_B), 0\}.
    \end{align*}
    That is, $\covec{w}_+$ and $\covec{w}_-$ are the positive and negative parts of $\covec{w}$ as usual, but with an additional $\frac{1}{4}$ added to each element where $y_A=y_B$. When $y_A\not= y_B$, we have $0 \leq w_\pm(y_A,y_B) \leq \frac{1}{4}$, but when $y_A = y_B = y$ we have $\frac{1}{4} \leq w_\pm(y,y) \leq \frac{1}{2}$. Therefore $w_+$ and $w_-$ both satisfy the conditions of Lemma \ref{lemma:2-point-range:nonsignaling}.
    
    Create the strategies $q_\pm\in\Hom^S_{NS}(Y,\{0,1\})$ according to Lemma \ref{lemma:2-point-range:nonsignaling}:
    \begin{align*}
        q_\pm(0,0\:|\:y_A,y_B) &= 1 + w_\pm(y_A,y_B) - w_\pm(y_A,y_A) - w_\pm(y_B,y_B)\\
        q_\pm(0,1\:|\:y_A,y_B) &= w_\pm(y_B,y_B) - w_\pm(y_A,y_B),\\
        q_\pm(1,0\:|\:y_A,y_B) &= w_\pm(y_A,y_A) - w_\pm(y_A,y_B),\\
        q_\pm(1,1\:|\:y_A,y_B) &= w_\pm(y_A,y_B).
    \end{align*}
    Then the stochastic matrices $Q_\pm$ associated to these strategies satisfy
    \begin{equation*}
        Q_+ - Q_- = \left(\begin{array}{c} \covec{w} - \covec{u} - \covec{v} \\ \covec{v}-\covec{w} \\ \covec{u}-\covec{w} \\ \covec{w} \end{array}\right).
    \end{equation*}
    Hence, $(Q_+ - Q_-)P = 0$ and $q_1\circ p = q_2\circ p$. But as $p$ is an epimorphism we have $Q_+ = Q_-$ and therefore $\covec{w} = \covec{0}$.
\end{proof}

Unfortunately, to extend the above proof to the case of hidden variables or quantum correlations one needs an analogue of Lemma \ref{lemma:2-point-range:nonsignaling} for these categories. However we know of no such results. To treat the classical and quantum cases, we first show that the left nullspace of a symmetric synchronous correlation decomposes into symmetric and antisymmetric parts. Then we prove that the symmetric nullspace and skew nullspace must independently vanish.

\begin{lemma}\label{lemma:epimorphism:decomposition}
    Let $p$ be a symmetric synchronous correlation with associated stochastic matrix $P$. Suppose $\covec{w}P = \covec{0}$. Then $v(y_A,y_B) = w(y_B,y_A)$ also has $\covec{v}P = \covec{0}$. In particular, if we decompose $\covec{w}$ into its symmetric and antisymmetric parts,
    \begin{align*}
        w^{(+)}(y_A,y_B) &= \tfrac{1}{2}(w(y_A,y_B) + w(y_B,y_A)) \text{ and}\\
        w^{(-)}(y_A,y_B) &= \tfrac{1}{2}(w(y_A,y_B) - w(y_B,y_A)),
    \end{align*}
    then $\covec{w}^{(\pm)}P = \covec{0}$.
\end{lemma}
\begin{proof}
    This follows from a straightforward calculation:
    \begin{align*}
        \sum_{y_A,y_B} v(y_A,y_B) p(y_A,y_B\:|\:x_A,x_B) 
        &= \sum_{y_A,y_B} w(y_B,y_A) p(y_A,y_B\:|\: x_A,x_B)\\
        &= \sum_{y_A,y_B} w(y_A,y_B) p(y_B,y_A\:|\: x_A,x_B)\\
        &= \sum_{y_A,y_B} w(y_A,y_B) p(y_A,y_B\:|\: x_B,x_A)\ =\ 0.
    \end{align*}
\end{proof}

In \cite{lackey2017nonlocal} we saw that there are nontrivial inequalities---even beyond symmetry---that must be satisfied to conclude that a nonsignaling correlation is classical. Nonetheless the proof of Theorem \ref{theorem:epimorphism:nonsignaling} did not require a full characterization of the elements of $\Hom^S_{NS}(X,\{0,1\})$. In the appendix we prove Lemma \ref{lemma:2-point-range:classical} that provides a sufficient condition for a correlation to be classical.

\begin{proposition}\label{proposition:epimorphisms:symmetric}
    Suppose $p\in\Hom^S_{NS}(X,Y)$ is epic and symmetric with associated stochastic matrix $P$ and suppose $\covec{w}P = \covec{0}$ with $w(y_A,y_B) = w(y_B,y_A)$. Then $\covec{w} = \covec{0}$.
\end{proposition}
\begin{proof}
    Exactly as in the proof of Theorem \ref{theorem:epimorphism:nonsignaling}, we construct
    \begin{align*}
        u(y_A,y_B) &= w(y_A,y_A) \text{ which is independent of $y_B$, and}\\
        v(y_A,y_B) &= w(y_B,y_B) \text{ which is independent of $y_A$,}
    \end{align*}
    and can conclude that $\covec{u}$ and $\covec{v}$ are also in the left nullspace of $P$ as this follows from $p$ being nonsignaling. Now however we rescale $\covec{w}$ so that $\|\covec{w}\|_\infty < \frac{1}{n^2}$ where $n = |X|$, and decompose this vector into positive and negative parts as
    \begin{align*}
        w_+(y_A,y_B) = \frac{1}{n}\indicator{y_A = y_B} + \max\{w(y_A,y_B), 0\},\\
        w_-(y_A,y_B) = \frac{1}{n}\indicator{y_A = y_B} + \max\{-w(y_A,y_B), 0\}.
    \end{align*}
    Then we have
    $$\sum_{y_A\not= y_B} w_\pm (y_A,y_B) < \frac{1}{n} \leq w_\pm(y_A,y_B),$$
    and
    $$\sum_{y_A} w_\pm(y_A,y_A) \leq 1 \leq 1 + \frac{1}{2}\sum_{y_A\not= y_B} w_{\pm}(y_A,y_B).$$
    These with the symmetry of $\covec{w}_\pm$ are precisely the conditions of Lemma \ref{lemma:2-point-range:classical}, and so we conclude that
    \begin{align*}
        q_\pm(0,0\:|\: y_A,y_B) &= 1 - w_\pm(y_A,y_A) - w_\pm(y_B,y_B) + w_\pm(y_A,y_B)\\
        q_\pm(0,1\:|\: y_A,y_B) &= w_\pm(y_B,y_B) - w_\pm(y_A,y_B),\\
        q_\pm(1,0\:|\: y_A,y_B) &= w_\pm(y_A,y_A) - w_\pm(y_A,y_B),\\
        q_\pm(1,1\:|\: y_A,y_B) &= w_\pm(y_A,y_B)\\
    \end{align*}
    define $q_\pm\in\Hom^S_{HV}(Y,\{0,1\})$. Now the remainder of the proof of Theorem \ref{theorem:epimorphism:nonsignaling} carries over verbatim.
\end{proof}

We know of no means to treat the skew part of the left nullspace by same sort of reasoning we have used in all the results above. Instead we provide an \textit{ad hoc} proof by directly building hidden variables strategies with the needed properties.

\begin{proposition}\label{proposition:epimorphisms:skew}
    Suppose $p\in\Hom^S_{NS}(X,Y)$ is epic and symmetric with associated stochastic matrix $P$ and suppose $\covec{v}P = \covec{0}$ with $v(y_A,y_B) = -v(y_B,y_A)$. Then $\covec{v} = \covec{0}$.
\end{proposition}
\begin{proof}
    If $\covec{v}\not= \covec{0}$ we may normalize this vector if necessary so that $\|\covec{v}\|_1 = \frac{2}{3}$. Decompose $\covec{v}$ into its usual positive and negative parts:
    $$v_{\pm}(y_A,y_B) = \max\{\pm v(y_A,y_B), 0\}.$$
    Note that as $\covec{v}$ is skew we have
    \begin{align}
        \label{equation:epimorphisms:skew-kernel1} v_{\pm}(y_A,y_B) &> 0 \text{ implies $v_\mp(y_A,y_B) = 0$, and}\\
        \label{equation:epimorphisms:skew-kernel2} v_{\pm}(y_A,y_B) &= v_{\mp}(y_B,y_A)
    \end{align}
    From (\ref{equation:epimorphisms:skew-kernel1}) we have
    \begin{align*}
        \tfrac{2}{3}\ =\ \|\covec{v}_1\|_1 &= \sum_{y_A,y_B} |v_+(y_A,y_B) - v_-(y_A,y_B)|\\
        &= \sum_{y_A,y_B} v_+(y_A,y_B) +  \sum_{y_A,y_B} v_-(y_A,y_B),
    \end{align*}
    and from (\ref{equation:epimorphisms:skew-kernel2}) we have
    $$\sum_{y_A,y_B} v_+(y_A,y_B) = \sum_{y_A,y_B} v_-(y_A,y_B).$$
    Hence
    $$\sum_{y_A,y_B} v_+(y_A,y_B) = \tfrac{1}{3} = \sum_{y_A,y_B} v_-(y_A,y_B).$$
    
    Define functions $\mu_\pm$ on $\Hom(Y,\{0,1,2\})$ as follows:
    \begin{equation}\label{equation:epimorphisms:skew-measure}
    \mu_{\pm}(f) = \left\{\begin{array}{cl}
        v_{\pm}(y_0,y_1) & \text{if $f(y_0) = 0$, $f(y_1) = 1$, but otherwise $f(y) = 2$;}\\
        \sum_{y'} v_\pm(y_0,y') & \text{if $f(y_0) = 1$, but otherwise $f(y) = 2$;}\\
        \sum_{y'} v_\pm(y',y_1) & \text{if $f(y_1) = 0$, but otherwise $f(y) = 2$;}\\
        0 & \text{in any other case.}\end{array}\right.
    \end{equation}
    Summing over $f\in\Hom(Y,\{0,1,2\})$, we see each of the three cases contributes the sum $\sum_{y_0,y_1} v_{\pm}(y_0,y_1) = \frac{1}{3}$. As $\mu_\pm$ are clearly nonnegative, these form probability distributions on $\Hom(Y,\{0,1,2\})$. Write $q_{\pm}\in\Hom^S_{HV}(Y,\{0,1,2\})$ for the classical synchronous correlations corresponding to $\mu_\pm$ respectively.
    
    Consider
    $$q_\pm(0,0\:|\: y_A,y_B) = \sum_f \mu_{\pm}(f)\indicator{0 = f(y_A)}\indicator{0 = f(y_B)};$$
    the only functions that contribute to this sum are those with $f(y_A) = f(y_B) = 0$. If $y_A \not= y_B$, then from (\ref{equation:epimorphisms:skew-measure}) we see there are no such functions in the support of $\mu_{\pm}$ and hence $q_\pm(0,0\:|\: y_A,y_B) = 0$. On the other hand, if $y_A = y_B$ then there are two types of contributions to this sum:
    \begin{enumerate}
        \item for each $y'\not= y_A$ the function that takes $f(y_A) = 0$, $f(y') = 1$, and otherwise has $f(y) = 2$ contributes $v_{\pm}(y_A,y')$, and
        \item the function with $f(y_A) = 0$ and $f(y) = 2$ for all $y\not= y_A$ contributes $\sum_{y'} v_{\pm}(y',y_A)$.
    \end{enumerate}
    And so we have
    $$q_\pm(0,0\:|\: y_A,y_A) = \sum_{y'\not= y_A} v_{\pm}(y_A,y') + \sum_{y'} v_{\pm}(y',y_A).$$
    Now using (\ref{equation:epimorphisms:skew-kernel2}) compute
    \begin{align*}
        &q_+(0,0\:|\: y_A,y_A) - q_-(0,0\:|\: y_A,y_A)\\
        &\qquad =\ \sum_{y'\not= y_a} v_+(y_A,y') + \sum_{y'} v_+(y',y_A) - \sum_{y'\not= y_A} v_-(y_A,y') - \sum_{y'} v_-(y',y_A)\\
        &\qquad =\ \left(\sum_{y'} v_+(y',y_A) - \sum_{y'\not= y_A} v_+(y',y_A)\right) - \left(\sum_{y'} v_-(y',y_A) - \sum_{y'\not= y_A} v_-(y',y_A)\right)\\
        &\qquad =\ v_+(y_A,y_A) - v_-(y_A,y_A) = -v(y_A,y_A)\ =\ 0.
    \end{align*}
    Clearly $q_+(0,0\:|\: y_A,y_B) - q_-(0,0\:|\: y_A,y_B) = 0$ for $y_A\not= y_B$ as both terms vanish.
    
    Now consider
    $$q_\pm(0,1\:|\: y_A,y_B) = \sum_f \mu_{\pm}(f)\indicator{0 = f(y_A)}\indicator{1 = f(y_B)},$$
    which has contributions from functions with $f(y_A) = 0$ and $f(y_B) = 1$. If $y_A = y_B$ there are no such functions. When $y_A \not= y_B$ there is only one such function in the support of $\mu_{\pm}$: the function that additionally has $f(y) = 2$ for all $y \not= y_A,y_B$. Then from (\ref{equation:epimorphisms:skew-measure}) we have $q_\pm(0,1\:|\: y_A,y_B) = v_{\pm}(y_A,y_B)$ and therefore
    $$q_+(0,1\:|\: y_A,y_B) - q_-(0,1\:|\: y_A,y_B) = v(y_A,y_B).$$
    Note this equation is also valid when $y_A = y_B$ as $v(y_A,y_A) = 0$.
    
    Finally consider
    $$q_\pm(0,2\:|\: y_A,y_B) = \sum_f \mu_{\pm}(f)\indicator{0 = f(y_A)}\indicator{2 = f(y_B)}.$$
    As above if $y_A = y_B$ there are no such functions that contribute this sum, while if $y_A\not= y_B$ there are a number of contributors:
    \begin{enumerate}
        \item for each $y'\not= y_A,y_B$ the function that takes $f(y_A) = 0$, $f(y') = 1$, and otherwise has $f(y) = 2$ contributes $v_{\pm}(y_A,y')$, and
        \item the function with $f(y_A) = 0$ and $f(y) = 2$ for all $y\not= y_A$ contributes $\sum_{y'} v_{\pm}(y',y_A)$.
    \end{enumerate}
    So for $y_A \not= y_B$, we note that $v_{\pm}(y_A,y_A) = 0$ and so
    $$q_\pm(0,2\:|\: y_A,y_B) = \sum_{y'\not= y_B} v_{\pm}(y_A,y') + \sum_{y'} v_{\pm}(y',y_A).$$
    Then just as we did before, we use (\ref{equation:epimorphisms:skew-kernel2}) to conclude
    \begin{align*}
        &q_+(0,2\:|\: y_A,y_B) - q_-(0,2\:|\: y_A,y_B)\\
        &\qquad =\ \sum_{y'\not= y_B} v_+(y_A,y') + \sum_{y'} v_+(y',y_A) - \sum_{y'\not= y_B} v_-(y_A,y') - \sum_{y'} v_-(y',y_A)\\
        &\qquad =\ \left(\sum_{y'} v_+(y',y_A) - \sum_{y'\not= y_B} v_+(y',y_A)\right) - \left(\sum_{y'} v_-(y',y_A) - \sum_{y'\not= y_B} v_-(y',y_A)\right)\\
        &\qquad =\ v_+(y_B,y_A) - v_-(y_B,y_A) = -v(y_A,y_B).
    \end{align*}
    Again this formula also holds for $y_A = y_B$.
    
    The other elements of the correlations $q_\pm$ can be analyzed in the same way, except for $q_\pm(2,2 \:|\: y_A,y_B)$ which can be determined from the fact that $q_\pm$ form probability distributions. The overall results are:
    \begin{align*}
        q_+(0,0 \:|\: y_A,y_B) - q_-(0,0 \:|\: y_A,y_B) &= 0\\
        q_+(0,1 \:|\: y_A,y_B) - q_-(0,1 \:|\: y_A,y_B) &= v(y_A,y_B)\\
        q_+(0,2 \:|\: y_A,y_B) - q_-(0,2 \:|\: y_A,y_B) &= -v(y_A,y_B)\\
        q_+(1,0 \:|\: y_A,y_B) - q_-(1,0 \:|\: y_A,y_B) &= -v(y_A,y_B)\\
        q_+(1,1 \:|\: y_A,y_B) - q_-(1,1 \:|\: y_A,y_B) &= 0\\
        q_+(1,2 \:|\: y_A,y_B) - q_-(1,2 \:|\: y_A,y_B) &= -v(y_A,y_B)\\
        q_+(2,0 \:|\: y_A,y_B) - q_-(2,0 \:|\: y_A,y_B) &= v(y_A,y_B)\\
        q_+(2,1 \:|\: y_A,y_B) - q_-(2,1 \:|\: y_A,y_B) &= v(y_A,y_B)\\
        q_+(2,2 \:|\: y_A,y_B) - q_-(2,2 \:|\: y_A,y_B) &= 0.\\
    \end{align*}
    Since $\covec{v}P = \covec{0}$ we have the associated stochastic matrices $(Q_+ - Q_-)P = 0$ and hence $q_+\circ p = q_-\circ p$. But as $p$ is epic, we conclude $q_+ = q_-$ and hence $\covec{v} = \covec{0}$ after all.
\end{proof}

Putting Propositions \ref{proposition:epimorphisms:symmetric} and \ref{proposition:epimorphisms:skew} together gives us the final result.

\begin{theorem}\label{theorem:epimorpism:symmetric}
    In $\category{FinSet}^S_{HV}$ and $\category{FinSet}^S_{Q}$, epimorphisms are precisely those correlations whose stochastic matrices have zero left kernel.
\end{theorem}

\end{section}

\begin{section}{Conclusions}

In most of the commonly studied categories sections and monomorphisms, and retractions and epimorphisms coincide and so there no ambiguity as to what it means for two objects to be isomorphic. Formally an isomorphism is a morphism that is both a section and a retractions. Therefore in any of $\category{FinSet}^S_{HV}$, $\category{FinSet}^S_Q$, or $\category{FinSet}^S_{NS}$ an isomorphism is just a bijective function as usual.

However in these categories monomorphism and epimorphism are strictly weaker conditions than section and retraction. This leads to a weaker notion of equivalence, the \emph{bimorphism}, which is morphism that is both epic and monic. In all of these categories, a bimorphism is a correlation whose stochastic matrix is nonsingular. In general, such a morphism is not an isomorphism as its inverse matrix will not be nonnegative, and hence not correspond to any correlation.

Fortunately, or unfortunately depending on one's perspective, this does not introduce a new notion of equivalence of finite sets (which would then have ramifications on arithmetic). For two sets to be bimorphic, they still must have the same cardinality.

\end{section}

%\bibliography{all.bib}

\appendix
\begin{section}{Boole's inequalities}

Understanding the relationships that must be satisfied by probabilities of overlapping sets dates back well into the 19th century. Today such inequalities go under the names of Boole \cite{boole1854investigation}, Bonferroni \cite{bonferroni1936teoria}, and Fr{\'e}chet \cite{frechet1935generalisation,frechet1960tableaux}, and are derived directly from the inclusion-exclusion principle. It is well recognized that these bounds have a deep relationship with Bell's inequality \cite{pitowsky1989george}, and it is in this context we examine them.

Originally, Boole considered the following problem: given two overlapping subsets $S_1$ and $S_2$ of some finite probability space $(\Omega,\Pr)$, then what can we say about $\Pr(S_1\cap S_2)$ given $\Pr(S_1)$ and $\Pr(S_2)$? Boole is attributed to showing
\begin{equation}\label{equation:boole:boole-inequalities}
    \max\{0,\Pr(S_1) + \Pr(S_2) - 1\} \leq \Pr(S_1\cap S_2) \leq \min\{\Pr(S_1),\Pr(S_2)\}.
\end{equation}
Clearly, $\Pr(S_1\cap S_2) \leq \Pr(S_1)$ and $\Pr(S_2)$ by monotonicity, and $\Pr(S_1\cap S_2) \geq 0$ by positivity, of the probability measure. The only remaining inequality
$$\Pr(S_1) + \Pr(S_2) - 1 \leq \Pr(S_1\cap S_2)$$
follows immediately from the inclusion-exclusion principle:
$$\Pr(S_1) + \Pr(S_2) - \Pr(S_1\cap S_2) = \Pr(S_1\cup S_2) \leq 1.$$

Our problem is similar: given a finite set $X = \{x_0, \dots, x_{n-1}\}$ and a collection of nonnegative numbers $w(x_j,x_k)$, when does there exist a probability space $(\Omega,\Pr)$ and subsets $S_j\subseteq\Omega$ so that
$$w(x_j,x_j) = \Pr(S_j) \text{ and } w(x_j,x_k) = \Pr(S_j\cap S_k).$$
Boole's inequalities provide necessary conditions:
\begin{enumerate}
    \item $w(x_j,x_k) \leq w(x_j,x_j)$,
    \item $w(x_j,x_k) \leq w(x_k,x_k)$, and
    \item $w(x_j,x_j) + w(x_k,x_k) \leq 1 + w(x_j,x_k)$.
\end{enumerate}
These are precisely the conditions (1--3) of Lemma \ref{lemma:2-point-range:nonsignaling}, which characterize when such numbers define a synchronous nonsignaling correlation. In addition to the obvious required symmetry $w(x_j,x_k) = w(x_k,x_j)$, there are further criteria that stem from relations between the probabilities of intersections of more than two subsets.

We claim that, without loss of generality, we may take $\Omega = \Hom(X,\{0,1\})$ where
$$S_j = \{f:X\to\{0,1\}\::\: f(x_j) = 1\}.$$
The complements of these sets we denote as
$$\overline{S}_j = \{f:X\to\{0,1\}\::\: f(x_j) = 0\},$$
and also introduce the notation $S_j^0 = \overline{S}_j$ and $S_j^1 = S_j$. Note that each $f\in\Hom(X,\{0,1\})$ has
$$\{f\} = \bigcap_{j=0}^{n-1} S_j^{f(x_j)}.$$
In particular, if some probability space $(\tilde{\Omega},\widetilde{\Pr})$ has $w(x_j,x_j) = \widetilde{\Pr}(\tilde{S}_j)$ and $w(x_j,x_k) = \widetilde{\Pr}(\tilde{S}_j\cap \tilde{S}_k)$ for subsets $\tilde{S}_j \subseteq \tilde{\Omega}$, then we define $F:\tilde{\Omega} \to \Omega$ by $F(\omega) = f$ for all $\omega \in \bigcap_{j=0}^{n-1} S_j^{f(x_j)}$. Then the push-forward measure $\Pr = F_*\widetilde{\Pr}$ also has $w(x_j,x_j) = \Pr(S_j)$ and $w(x_j,x_k) = \Pr(S_j\cap {S}_k)$.

Fix $\Pr$, a probability measure on $\Omega$. The we associate to this measure two vectors $\vec{p},\vec{w}\in \mathbb{R}^{2^n}$. Writing $j\in\{0,\dots,2^n-1\}$ in bits,
$$j = j_0 + j_1 2 + \cdots + j_{n-1} 2^{n-1},$$
the entries of these vectors at position $j$ are defined as
$$p_j = \Pr\left(\bigcap_{k=0}^{n-1} S_k^{j_k}\right) \text{ and } w_j = \Pr\left(\bigcap_{k\::\: j_k=1} S_k\right).$$
For example, when $n=2$ these are
$$\vec{p} = \left(\begin{array}{c}
\Pr(\overline{S}_0 \cap \overline{S}_1)\\
\Pr(\overline{S}_0 \cap S_1)\\
\Pr(S_0 \cap \overline{S}_1)\\
\Pr(S_0 \cap S_1)
\end{array}\right) \text{ and } 
\vec{w} = \left(\begin{array}{c}
\Pr(\Omega)\\
\Pr(S_1)\\
\Pr(S_0)\\
\Pr(S_0 \cap S_1)
\end{array}\right).$$
These vectors are related according to the following result, see also \cite{fontana2017characterization}.

\begin{lemma}
    We have $\scriptsize{\left(\begin{array}{cc} 1 & 1\\ 0 & 1\end{array}\right)}^{\otimes n} \vec{p} = \vec{w}$.
\end{lemma}
\begin{proof}
    Define the matrix
    $$M_\ell = \openone \otimes \cdots \otimes \openone \otimes \left(\begin{array}{cc} 1 & 1\\ 0 & 1\end{array}\right) \otimes \openone \otimes \cdots \openone$$
    where the nonidentity matrix is in the $\ell$th factor. The $(j,k)$ entry of this matrix is
    $$[M_\ell]_{j,k} = \indicator{j_0 = k_0}\cdots \indicator{j_{\ell-1} = k_{\ell-1}}(\indicator{j_\ell = k_\ell} + \indicator{j_\ell = 0}\indicator{k_\ell = 1}) \indicator{j_{\ell+1} = k_{\ell+1}}\cdots \indicator{j_{n-1} = k_{n-1}}.$$
    Then
    $$[M_\ell]_{j,k}p_k = \left\{\begin{array}{cl}
        \Pr(\bigcap_{m\not=\ell} S_m^{j_m}\cap\bar{S}_\ell) & \text{if $j_m = k_m$ (for $m\not=\ell$) and $j_\ell = 0$ and $k_\ell = 1$,}\\
        \Pr(\bigcap_{m\not=\ell} S_m^{j_m}\cap S_\ell) & \text{if $j=k$, and}\\
        0 & \text{otherwise.}
    \end{array}\right.$$
    So if $j_\ell = 0$ then
    $$\sum_k [M_\ell]_{j,k}p_k = \Pr(\bigcap_{m\not=\ell} S_m^{j_m}\cap\bar{S}_\ell) + \Pr(\bigcap_{m\not=\ell} S_m^{j_m}\cap S_\ell) = \Pr(\bigcap_{m\not=\ell} S_m^{j_m}).$$
    While if $j_\ell = 1$ then
    $$\sum_k [M_\ell]_{j,k}p_k = \Pr(\bigcap_{m\not=\ell} S_m^{j_m}\cap S_\ell).$$
    That is, when $j_\ell = 0$ the $\ell$th term is removed from the intersection while if $j_\ell = 0$ then it is kept, and is $S_\ell$. But $\scriptsize{\left(\begin{array}{cc} 1 & 1\\ 0 & 1\end{array}\right)}^{\otimes n} = M_1\cdots M_{n-1}$, and so after applying all $M_\ell$ the $j$th entry of the result is $\Pr(\bigcap_{\ell\::\: j_\ell = 1} S_\ell) = w_j$.
\end{proof}

The lemma provides the following characterization: for any vector $\vec{w}\in\mathbb{R}^{2^n}$, it entrywise satisfies
\begin{equation}\label{equation:Boole:gen-boole-inequality}
\left(\begin{array}{cc} 1 & -1\\ 0 & 1\end{array}\right)^{\otimes n} \vec{w} \geq \vec{0}
\end{equation}
if and only if there exists a probability measure $\Pr$ on $\Omega$ so that $w_j = \Pr(\bigcap_{k:j_k=1} S_k)$, and when this is the case the measure is unique. Unfortunately this does not solve our problem as our vector $\vec{w}$ is only defined for entries $w_j$ where $j$ has bit-Hamming weight at most two.

One question to ask is under what conditions on $w(x_j,x_k)$ can we extend these to a whole vector that satisfies inequalities (\ref{equation:Boole:gen-boole-inequality}). For example, in the case $n=3$ these inequalities read
$$\left(\begin{array}{cccccccc}
1 & 1 & -1 &  1 & -1 &  1 &  1 & -1\\
0 & 1 &  0 & -1 &  0 & -1 &  0 &  1\\
0 & 0 &  1 & -1 &  0 &  0 & -1 &  1\\
0 & 0 &  0 &  1 &  0 &  0 &  0 & -1\\
0 & 0 &  0 &  0 &  1 & -1 & -1 &  1\\
0 & 0 &  0 &  0 &  0 &  1 &  0 & -1\\
0 & 0 &  0 &  0 &  0 &  0 &  1 & -1\\
0 & 0 &  0 &  0 &  0 &  0 &  0 &  1
\end{array}\right)
\left(\begin{array}{c}
1 \\ w(x_2,x_2) \\ w(x_1,x_1) \\ w(x_1,x_2) \\ w(x_0,x_0) \\ w(x_0,x_2) \\ w(x_1,x_2) \\ w_7
\end{array}\right) \geq
\left(\begin{array}{c}
0 \\ 0 \\ 0 \\ 0 \\ 0 \\ 0 \\ 0 \\ 0
\end{array}\right).$$
Isolating $w_7$ we obtain
\begin{align}\label{equation:3-boole-inequality}
&\max\left\{\begin{array}{l} 
0,\\ 
-w(x_0,x_0) + w(x_0,x_1) + w(x_0,x_2),\\
-w(x_1,x_1) + w(x_0,x_1) + w(x_1,x_2)\\
-w(x_2,x_2) + w(x_0,x_2) + w(x_1,x_2)
\end{array}\right\} \leq w_7\\
&\qquad\qquad \leq
\min\left\{\begin{array}{l} 
1 - w(x_0,x_0) - w(x_1,x_1) - w(x_2,x_2)\\
\qquad +\ w(x_0,x_1) + w(x_0,x_2) + w(x_1,x_2),\\
w(x_0,x_1),\\
w(x_0,x_2),\\
w(x_1,x_2)
\end{array}\right\}.\nonumber
\end{align}
As (putatively) $w_7 = \Pr(S_1\cap S_2\cap S_3)$ these inequalities can be viewed as the direct generalization of Boole's inequalities to threefold intersections. In order that $w_7$ exists, we must have that each of the four bounded terms on the left are no greater than the each of the bounding terms on the right. This provides inequalities that are satisfied if and only if $w_7$, and hence a probability measure, exists. Note that these $16$ inequalities are precisely those found to generate synchronous Bell inequalities \cite[Equations (6,7)]{lackey2017nonlocal}, albeit with a different index convention.

Extending to four-fold intersections and beyond does not follow easily, unlike the triple-intersection case above. To conclude this section we turn to an easier problem, and merely provide sufficient conditions for the existence of a probability measure. The most straightforward case appears to characterize solutions to (\ref{equation:Boole:gen-boole-inequality}) when $w_j = 0$ for bit-Hamming weight three or higher, although other simple examples likely exist.

\begin{lemma}\label{lemma:2-point-range:classical}
Let $X = \{x_0, \dots, x_{n-1}\}$ be a finite set and $w(x_j,x_k)$ be nonnegative values that satisfy
\begin{enumerate}
    \item $w(x_j,x_k) = w(x_k,x_j)$,
    \item $w(x_j,x_j) \geq \sum_{k\not=j} w(x_j,x_k)$, and
    \item $\sum_{j=0}^{n-1} w(x_j,x_j) \leq 1 + \sum_{j<k} w(x_j,x_k)$.
\end{enumerate}
Then there exists a probability measure on $\Hom(X,\{0,1\})$ such that
$$w(x_j,x_k) = \Pr(\{f\::\: f(x_j) = f(x_k) = 1\}).$$
That is,
\begin{align*}
q(0,0\:|\:x_A,x_B) &= 1 - w(x_A,x_A) - w(x_B,x_B) + w(x_A,x_B)\\
q(0,0\:|\:x_A,x_B) &= w(x_B,x_B) - w(x_A,x_B)\\
q(0,0\:|\:x_A,x_B) &= w(x_A,x_A) - w(x_A,x_B)\\
q(0,0\:|\:x_A,x_B) &= w(x_A,x_B)
\end{align*}
defines a synchronous hidden variables correlation.
\end{lemma}
\begin{proof}
    As indicated above, we define $\vec{w}$ as follows:
    $$w_j = \left\{\begin{array}{cl}
        1 & \text{if $j = 0$,}\\
        w(x_\ell,x_\ell) & \text{if $j = 2^\ell$,}\\
        w(x_\ell,x_m) & \text{if $j = 2^\ell + 2^m$ (for $\ell \not= m$),}\\
        0 & \text{otherwise.}\end{array}\right.$$
    Similar to the previous lemma we define a matrix
    $$N_i = \openone \otimes \cdots \otimes \openone \otimes \left(\begin{array}{cc} 1 & -1\\ 0 & 1\end{array}\right) \otimes \openone \otimes \cdots \openone$$
    where again the nonidentity matrix is in the $i$th factor, for which the $(j,k)$ entry is
    $$[N_i]_{j,k} = \indicator{j_0 = k_0}\cdots \indicator{j_{i-1} = k_{i-1}}(\indicator{j_i = k_i} - \indicator{j_i = 0}\indicator{k_i = 1}) \indicator{j_{i+1} = k_{i+1}}\cdots \indicator{j_{n-1} = k_{n-1}}.$$
    Then on any vector $\vec{v}$ supported on indices of bit-Hamming weight at most two,
    $$[N_i]_{j,k}v_k = \left\{\begin{array}{cl}
        -v_{j + 2^i} & \text{if $j_m = k_m$ (for $m\not=i$) and $j_i = 0$ and $k_i = 1$,}\\
        v_j & \text{if $j=k$, and}\\
        0 & \text{otherwise.}
    \end{array}\right.$$
    Note that if $j$ has bit-Hamming weight three or larger then $[N_i]_{j,k}v_k = 0$ for all $k$, and so $[N_i\vec{v}]_j = 0$.
    
    When $j$ has bit-Hamming weight two then $[N_i]_{j,k}v_k = 0$, unless $j=k$ for which it equals $v_j$. That is $[N_i\vec{v}]_j = v_j$.
    
    When $j$ has bit-Hamming weight one, $j = 2^\ell$, we have
    $$[N_i]_{j,k}v_k = \left\{\begin{array}{cl}
        -v_{2^i + 2^\ell} & \text{if $k = 2^i + 2^\ell$ where $i\not= \ell$,}\\
        v_{2^\ell} & \text{if $k = 2^\ell$, and}\\
        0 & \text{otherwise.}
    \end{array}\right.$$
    Thus $[N_i\vec{v}]_j = v_{2^\ell} - v_{2^i + 2^\ell}$ when $i \not = \ell$ and $[N_i\vec{v}]_j = v_{2^\ell}$ when $i = \ell$.
    
    Finally, when $j = 0$ we have
    $$[N_i]_{0,k}v_k = \left\{\begin{array}{cl}
        -v_{2^i} & \text{if $k = 2^i$,}\\
        v_0 & \text{if $k = 0$, and}\\
        0 & \text{otherwise}
    \end{array}\right.$$
    and so $[N_i\vec{v}]_0 = v_0 - v_{2^i}$.
    
    Putting all these results together, we find
    \begin{equation}\label{equation:Boole:the-N}
        [N_i\vec{v}]_j = \left\{\begin{array}{cl}
            v_0 - v_{2^i} & \text{if $j = 0$,}\\
            v_{2^i} & \text{if $j = 2^i$,}\\
            v_{2^\ell} - v_{2^i + 2^\ell} & \text{if $j = 2^\ell$ for $\ell\not= i$,}\\
            v_{2^\ell + 2^m} & \text{if $j = 2^\ell + 2^m$ for $\ell\not= m$,}\\
            0 & \text{otherwise.}
            \end{array}\right.
    \end{equation}
    We aim to compute
    $$\scriptsize{\left(\begin{array}{cc} 1 & -1\\ 0 & 1\end{array}\right)}^{\otimes n} \vec{w} = \left(\prod_i N_i\right)\vec{w}$$
    and verify the inequalities (\ref{equation:Boole:gen-boole-inequality}).
    
    We claim
    \begin{equation}\label{equation:Boole:the-final-answer}
        \left[ \prod_{i=0}^{n-1} N_i \vec{w}\right]_j = \left\{\begin{array}{cl}
            1 - \sum_{\ell=0}^{n-1} w_{2^\ell} + \sum_{0 \leq \ell < m \leq n-1} w_{2^\ell + 2^m} & \text{if $j = 0$,}\\
            w_{2^\ell} - \sum_{m=0}^{\ell-1} w_{2^\ell + 2^m} - \sum_{m=\ell+1}^{n-1} w_{2^\ell + 2^m} & \text{if $j = 2^\ell$,}\\
            w_{2^\ell + 2^m} & \text{if $j = 2^\ell + 2^m$ for $\ell\not= m$,}\\
            0 & \text{otherwise,}
            \end{array}\right.
    \end{equation}
    which we prove by induction. Trivially, $\vec{w}$ starts in this form, so inductively assume
    \begin{equation*}
        \left[ \prod_{i=0}^{n-2} N_i \vec{w}\right]_j = \left\{\begin{array}{cl}
            1 - \sum_{\ell=0}^{n-2} w_{2^\ell} + \sum_{0 \leq \ell < m \leq n-2} w_{2^\ell + 2^m} & \text{if $j = 0$,}\\
            w_{2^\ell} - \sum_{m=0}^{\ell-1} w_{2^\ell + 2^m} - \sum_{m=\ell+1}^{n-2} w_{2^\ell + 2^m} & \text{if $j = 2^\ell$,}\\
            w_{2^\ell + 2^m} & \text{if $j = 2^\ell + 2^m$ for $\ell\not= m$,}\\
            0 & \text{otherwise.}
            \end{array}\right.
    \end{equation*}
    Now for $j$ of weight two
    $$\left[ N_{n-1}\prod_{i=0}^{n-2} N_i \vec{w}\right]_{2^\ell + 2^m} = w_{2^\ell + 2^m}$$
    directly from the $j = 2^\ell + 2^m$ line of (\ref{equation:Boole:the-final-answer}). For $j$ of weight one there are two cases to consider. First
    $$\left[ N_{n-1}\prod_{i=0}^{n-2} N_i \vec{w}\right]_{2^{n-1}} = w_{2^{n-1}} - \sum_{m=0}^{n-2} w_{2^\ell + 2^m}$$
    again directly from the $j = 2^{n-1}$ line of (\ref{equation:Boole:the-final-answer}). For $\ell < n-1$ we use the $j = 2^\ell$ line of (\ref{equation:Boole:the-final-answer}) to obtain
    $$\left[ N_{n-1}\prod_{i=0}^{n-2} N_i \vec{w}\right]_{2^\ell} = w_{2^\ell} - \sum_{m=0}^{\ell-1} w_{2^\ell + 2^m} - \sum_{m=\ell+1}^{n-2} w_{2^\ell + 2^m} - w_{2^\ell + 2^{n-1}}.$$
    Finally, for $j=0$ we have from (\ref{equation:Boole:the-final-answer}) that
    \begin{align*}
        \left[ N_{n-1}\prod_{i=0}^{n-2} N_i \vec{w}\right]_0 &= 1 - \sum_{\ell = 0}^{n-2} w_{2^\ell} + \sum_{0 \leq \ell < m \leq n-2} w_{2^\ell + 2^m} - w_{2^{n-1}} + \sum_{m=0}^{n-2} w_{2^\ell + 2^m}\\
        &= 1 - \sum_{\ell=0}^{n-1} w_{2^\ell} + \sum_{0 \leq \ell < m \leq n-1} w_{2^\ell + 2^m}.
    \end{align*}

    Now, $\left(\prod_i N_i\right)\vec{w} \leq 0$ directly from this formula and the assumptions on the coordinates $\vec{w}$ made in the statement.
\end{proof}

\end{section}

\end{document}